\documentclass[smallextended,natbib]{svjour3}
\smartqed  
\usepackage{amsmath,amsfonts,amssymb}
\usepackage{graphicx,mathrsfs,hyperref,xcolor}

\def \ie {i.e.~}

\def \RR {\mathbb{R}}
\def \F {\mathcal{F}}
\def \DUMMY {\tilde V}
\newcommand{\review}[1]{\textcolor{black}{#1}}
\newcommand{\reviewd}[1]{\textcolor{black}{#1}}

\newcommand{\NomProblem}[1]{B$k$-WWC\,}
\begin{document}

\title{An Exact Approach for the Balanced \review{k-Way} Partitioning \review{Problem with Weight Constraints and its Application to Sports Team Realignment} 
\thanks{A preliminary version of this paper appeared at ISCO 2016}}
\thanks{A preliminary version of this paper appeared at ISCO 2016}
\author{Diego Recalde \and Daniel Sever\'in \and \mbox{} Ramiro Torres \and Polo Vaca}

\institute{D. Recalde \and R. Torres \and P. Vaca \at Escuela Polit\'ecnica Nacional, Departamento de Matem\'atica, Ladr\'on de Guevara E11-253, EC170109 Quito, Ecuador. \\ \email{\{diego.recalde, ramiro.torres, polo.vaca\}@epn.edu.ec}\and D. Recalde \at Research Center on Mathematical Modelling (MODEMAT), Escuela Polit\'ecnica Nacional, Quito, Ecuador.  \and D. Sever\'in \at FCEIA, Universidad Nacional de Rosario and CONICET, Argentina.\\ \email{daniel@fceia.unr.edu.ar}}

\date{Received: date / Accepted: date} 

\maketitle

\begin{abstract}

In this work a balanced $k$-way partitioning problem with weight constraints  is defined to model \review{the} sports team realignment. Sports teams must be partitioned into a fixed number of groups according to some regulations, where the total distance of the road trips that all teams must travel to play a Double Round Robin Tournament in each group is minimized. \review{Two integer programming formulations  for this problem are introduced, and the validity of three families of inequalities associated to the polytope of these formulations is proved. The performance of a tabu search procedure and a Branch \& Cut algorithm, which uses the valid inequalities as cuts, is evaluated over simulated and real-world instances.} In particular, an optimal solution
for the realignment of the Ecuadorian Football league is reported and the methodology can be suitable adapted for the realignment of other sports leagues. 
   
\end{abstract}
\section{Introduction}
 \label{introduction}

\review{A fundamental problem in Combinatorial Optimization is to partition a graph in several parts. There are many different versions of graph partitioning problems depending on the number of parts required, the type of weights on the edges or nodes, and the inclusion of several other constraints like restricting the number of nodes in each part. Usually, the objective of these problems is to divide the set of nodes into subsets with a strong internal connectivity and a weak connectivity between them. Most versions of these problems are known to be NP-hard. In this paper, a problem consisting of partitioning a complete or a general graph in a fixed number of subsets of nodes such that their cardinality differ at most in one and the total weight of each subset is bounded, is introduced. The objective aims to minimize the total cost of edges with end-nodes in the same subset. The motivation to state this problem was the realignment of the second category of the Ecuadorian football league.}

The sports team realignment deals with partitioning a fixed number of professional sports teams into groups or divisions of similar characteristics, in which a tournament is played. Commonly, a geographical criterion is used to construct the divisions in order to minimize the sum of intra-divisional travel costs. The  Ecuadorian football federation (FEF) adapted the last benchmark and imposes that the realignment of the  second category of the professional football league must be made by considering provinces instead of teams. Thus, in the realignment, the provinces are divided into four geographical zones. In each zone, two \emph{subgroups} of teams are randomly constructed by regarding the two best teams of each province, satisfying that two teams of the same province do not belong to the same subgroup, and every subgroup must have the same number of best and second-best teams whenever possible. The eight subgroups play a Double Round Robin Tournament, where the champion of each subgroup and the four second-best teams  with the highest scores advance to another stage of the championship. In 2014, FEF managers requested to the authors of this work how to make the realignment of the second category of the Ecuadorian football league in an optimal way, \review{considering 21 provinces and 42 teams to be partitioned in 4 groups.}   \review{Prior to this requirement, they made a realignment in which the total road travel distance was 39830.4 km. However, the optimal solution reduced this road travel distance in 1532.4 km with the corresponding highly relevant logistics and economical benefits for the large number of people involved. A graphical representation  is depicted in Figure \ref{figura_solucion_inicial}. In a preliminary work \cite{ISCO2016}, the problem was modeled as a  $k$-clique partitioning problem with constraints on the sizes and weights of the cliques and formulated as an integer program.  Unfortunately,  there was an instance related to a proposal to change the design of the realignment  that could not be solved to optimality. Consequently, the present paper aims to improve the previous practical results and provide more theoretical insightful about this problem.} 

\begin{figure}[!htb] \centering
\includegraphics[scale=0.85]{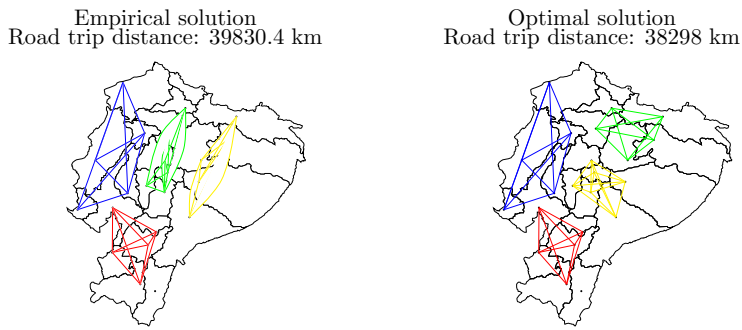}
\caption{Empirical vs optimal solution, 2014 edition}
\label{figura_solucion_inicial}
\end{figure}


The proposal for realignment the second category of the Ecuadorian football league addressed in the previous work, and revisited in this paper, consists of directly making the divisions by using teams instead of provinces, as it is done in other international leagues. During the 2015 season of the second category of the Ecuadorian football league, 44 teams participated (22 provincial associations) and the set of teams was divided into $4$ groups with $5$ teams and $4$ groups with $6$ teams. In the context of this problem, the distance between two teams is defined as the road trip distance between the venues of the teams. The strongness or weakness of a team is quantified by means of a parameter that measures its football level considering the historical performance of each team. Thus, the problem studied in this paper consists of partitioning the set of teams into $k$ groups such that \review{ the number of teams in each group differs at most by one,  there exists a certain homogeneity of football performance among the teams in each class of the partition, and minimizing an objective function corresponding to the total geographical distance intra-group.}

The sports team realignment problem has been modeled in different ways and for different leagues in other countries.  A quadratic binary programming model is set up to divide 30 teams, of the National Football League (NFL) in the United States, into 6 compact divisions of 5 teams each \citep{Saltzman1996}. The results, obtained directly from a nonlinear programming solver, are considerably less expensive for the teams in terms of total intra-divisional travel, in comparison with the realignment of the 1995 edition of this league. On the other hand, \citet{MacDonaldPulleyblank2014}  propose a
geometric method to construct realignments for several sports leagues
in the United States: NHL, MLB, NFL and NBA. The authors claim that
with their approach they always find the optimal solution. To prove
this, they solve mixed integer programming problems
corresponding to practical instances, using CPLEX.

In the case that it is possible to divide the teams into divisions of equal size,
this problem can be modeled as a $k$\textit{-way equipartition problem}: given an undirected graph with $n$ nodes and edge costs, the problem consists of finding a $k$-partition of the set of nodes, each of the same size,
such that the total cost of edges which have both endpoints in one of
the subsets of the partition is minimized. An example of this case was given by \citet{Mitchell2001}, who optimally solved the realignment of the
NFL in the United States for $32$ teams and $8$ divisions by using a branch-and-cut
algorithm. Moreover, the author shows that the 2002 edition of the NFL could have reduced the sum of intra-divisional travel distances by $45\%$.

When $n\!\mod k \neq 0$, the sports team realignment problem can be modeled as a \textit{Clique Partitioning Problem} (CPP) with constraints in the sizes of the cliques \citep{Mitchell2007},
as we will see in the next section.
The CPP has been extensively studied in the literature. This graph optimization problem was introduced by \citet{GROTSCHEL}
to formulate a clustering problem. They studied it from a polyhedral point of view and its theoretical results have been used in
a cutting plane algorithm that includes heuristic separation routines for some classes of facet-defining inequalities.
\citet{Ferreira} analyzed the problem of partitioning a graph satisfying capacity constraints on the sum of the node weights
in each subset of the partition.
\citet{Jaehn} proposed a branch-and-bound algorithm where tighter bounds for each node in the search tree are reported.
Additionally, a variant where constraints on the size of cliques are introduced for the CPP is studied by \citet{LABBE}
and a competitive branch-and-cut for a particular problem based on it have been implemented \citep{LABBE2}. \review{Regarded to applications of the Graph Partitioning Problem, it is widely known that the canonical application of this problem is the distribution of work to processors of a parallel machine \citep{HENDRICKSON_2000}. Other well known applications include VLSI circuit design \citep{VLSI_2011}, Mobile Wireless Communications \citep{Fairbrother_EtAl_2017} and the already mentioned sports team realignment \citep{Mitchell_2003}. For a complete survey of applications and recent advances in graph partitioning, see \cite{Buluc2016}. }

\review{This paper proposes two integer formulations for the $k$\textit{-way equipartition problem} which are formally defined in Section \ref{formulacion}. Moreover, valid inequalities for the polytopes of these formulations are derived in Section \ref{valid_inequalities}, which are  integrated in a Branch \& Bound \& Cut scheme to solve to optimality instances of 54 teams and, in particular, a hard real-world instance
of 44 teams not solved in a previous work}. Additionally, a tabu search method for finding feasible high quality solutions is shown in Section \ref{SectionHeuristic}. In Section \ref{SectionPracticalExperience}, the tabu search method and the usage of valid inequalities are integrated in several strategies to solve the real-world as well as the simulated instances. Finally, 
Section \ref{Conclusions} concludes the paper with some remarks.

\section{Problem definition and integer programming formulations}
\label{formulacion}
By associating the venues of teams with the nodes of a graph, the distance between venues with costs on the edges, and football performance
of the teams with weights on the nodes, \reviewd{a typical realignment problem} can be modeled as a $k$-way partitioning
problem with constraints on the size (number of nodes in each subset differs at most in one) and weight of the cliques (total sum of node
weights in the clique).  From now on, we refer to this problem as a \emph{balanced} $k$\emph{-\reviewd{way partitioning
problem with weight constraints} (B$k$-WWC)}.

Let $G=(V,E)$ be an undirected complete graph with node set
$V=\{1,\ldots,n\}$, edge set $E=\left\{\{i,j\}:i,j\in V, i \neq j\right\}$, cost
on the edges $d: E \longrightarrow \RR^+$,
weights on nodes $w:V \longrightarrow \RR^+$ and a fixed number
$k \geq 2$.  A $k$-\emph{partition} of $G$ is a collection of $k$ subgraphs $(V_1,E(V_1)),\ldots,(V_k,E(V_k))$ of $G$,
where $V_i \neq \emptyset$ for all $i=1,\dots, k$, $V_i\cap V_j= \emptyset$ for all $i\neq j$,
$\cup_{i=1}^k V_i=V$, and $E(V_i)$ is the set of edges with end nodes in $V_i$. \review{Note that all subgraphs $(V_i,E(V_i))$ are cliques since $G$ is a complete graph.} 
Moreover, let $W_L, W_U\in \RR^+$, $W_L\leq W_U$, be the lower and
upper bounds, respectively, for the weight of each clique (which is part of the input of our problem). The weight
of a clique is the total sum of the node weights in the
clique. Then, B$k$-WWC consists of finding a $k$-way
partition such that
\begin{align}
& W_L \leq \sum_{j\in V_c}w_j \leq W_U,& \quad & \forall \;
  c=1,\ldots,k,\label{cond_pesos}\\
&  \left| |V_i|-|V_j| \right| \leq 1,& \quad & \forall \; i,j =
  1,2,\ldots,k,\quad i< j,\label{cond_nodos}
\end{align}
and the total edge cost over all cliques is minimized.
In a previous work \citep{ISCO2016}, the NP-hardness of the B$k$-WWC was proved by a polynomial transformation from the 3-Equitable Coloring Problem.

It is known that algorithms based on integer programming techniques are proved to be the best tools for dealing with problems such as B$k$-WWC.
As it was mentioned in the introduction, another CPP was successfully addressed by \citet{LABBE2}, which used an integer programming formulation for the size-constrained clique partitioning problem given by \citet{LABBE}. In that formulation, a binary variable is defined for each edge. Let $x_{ij}$ be the variable associated to the edge $\{i,j\}$. Then, $x_{ij}=1$
if nodes $i$ and $j$ belong to the same clique, and $x_{ij}=0$ otherwise. The formulation is stated as follows:
\begin{align}
& \min \sum_{\{i,j\}\in E} d_{ij} x_{ij} \label{obj_1}\\
& \mbox{s.t.} \nonumber\\
& \hspace{0.4cm} x_{ij} + x_{jl} - x_{il} \leq 1,& \forall \; 1\leq i < j< l\leq n, \label{restr_11}\\
& \hspace{0.4cm} x_{ij} -x_{jl} +x_{il} \leq 1,& \forall \; 1\leq i < j< l\leq n,  \label{restr_12}\\
& - x_{ij} +x_{jl} + x_{il} \leq 1,& \forall \; 1\leq i < j< l\leq n,  \label{restr_13}\\
& F_L \leq 1 + \sum_{j:\{i,j\}\in E} x_{ij} \leq F_U,&\qquad \forall \; i\in V, \label{restr_2}\\
& \hspace{0.4cm} x_{ij} \in \{0,1\}, & \forall\; \{i,j\}\in E \label{restr_4}
\end{align}
The objective function \eqref{obj_1} seeks to minimize the total edge cost of the $k$ cliques.
Constraints \eqref{restr_11}, \eqref{restr_12}, and \eqref{restr_13} are the so-called triangle inequalities introduced by \citet{GROTSCHEL},
which guarantee that if three nodes $i,j,l$ of $V_c$ are linked by two edges $\{i,j\}$ and $\{j,l\}$, then the third edge $\{i,l\}$ must be
also included in the solution, for all $c$.
Constraints \eqref{restr_2} ensure that the cardinality of each clique is between values $F_L$ and $F_U$, and constraints \eqref{restr_4}
determine that all variables are binary.

In our case, in order to model the B$k$-WWC, the cardinality of any two cliques must differ at most in one, i.e., from now on, $F_L := \left \lfloor n/k \right \rfloor\geq 2$ and $F_U := \left \lceil n/k \right \rceil$.  
Moreover, additional constraints that impose the weight requirements on each clique are included:
\begin{equation}
W_L \leq  w_i + \sum_{j:\{i,j\}\in E} w_j x_{ij} \leq W_U, \qquad \forall\; i\in V, \label{restr_3}
\end{equation}

As the cardinality of each subset in the partition depends on $n$ and $k$, when $k$ divides $n$, the formulation \eqref{obj_1}-\eqref{restr_3}
returns $k$ cliques with exactly $n/k$ nodes and the problem corresponds  to the $k$-way equi-partition problem \cite{Mitchell2001}.
In this case, constraints (\ref{restr_2}) can be rewritten as:
\begin{equation}
\sum_{j:\{i,j\}\in E} x_{ij} = \frac{n}{k} - 1, \qquad \forall \; i\in V, \label{restr_2x}
\end{equation}

When $k$ does not divide $n$, the previous constraints are not enough to guarantee that integer solutions represent partitions  of $k$ cliques. \review{In fact, observe that  formulation \eqref{obj_1}-\eqref{restr_3} may admit feasible solutions for different values of $k$; for example, consider as an instance of B$k$-WWC a complete graph with $n=23$ and $k=7$, which implies that $F_L=3$ and $F_U=4$. However, a feasible solution for this instance could be a partition consisting of one subset of nodes with cardinality equal to $3$ and five subsets with cardinality equal to $4$, which corresponds to a non desired value of $k=6$.}

Two alternatives are proposed to overcome this issue. On the one hand, dummy nodes are added to the graph until the condition $n\!\mod k = 0$ is met, \ie a set of dummy nodes $\DUMMY$ of cardinality $k-(n\mod k)$ is defined. Moreover, for all $i\in \DUMMY$, costs $d_{ij}=0$ for
$j \in V \cup \DUMMY$, and weights $w_i = 0$ are fixed. After that, $V$ must be updated with $V\cup \DUMMY$, and the same for $E$. 
Finally, observe that two dummy nodes must not be assigned to the same clique in the partition. In order to avoid this impasse, the following constraint is considered:
\begin{equation}
\sum_{\{i,j\} \subset \DUMMY} x_{ij}=0
\label{rest_ceros}
\end{equation} 

We call $\F_1$ to the formulation of B$k$-WWC  composed by the objective function \eqref{obj_1} and constraints \eqref{restr_11}-\eqref{restr_13}, \eqref{restr_4}-\eqref{rest_ceros}.

\review{On the other hand}, an alternative to the inclusion of dummy nodes (as formulation $\F_1$ \reviewd{states}) is to consider a new
constraint as follows. Note that in a balanced partition of a graph, there exists $r\in \mathbb{N}$ and $k-r$ disjoint subsets of cardinality
$\left \lceil n/k \right \rceil$ and $\left \lfloor n/k \right \rfloor$, respectively. \reviewd{That is, $n=r F_U + \left(k-r \right) F_L$
where $r= n-k F_L = n \!\mod k$}.
\reviewd{Note also that the total number of edges} is $r F_U (F_U - 1)/2 + (k - r) F_L(F_L-1)/2$. Let $\beta_{n,k}$ be the last number. It is easy to see that $\beta_{n,k} = \beta_{n,k'}$ implies
$k = k'$, and therefore the following equality forces the partition to have size $k$:
\begin{equation}
 \sum_{\{i,j\}\in E}  x_{ij} = \beta_{n,k}
 \label{restr_5}
\end{equation}
We call $\F_2$ to the formulation composed by the objective function \eqref{obj_1} and constraints \eqref{restr_11}-\eqref{restr_3} and \eqref{restr_5}.

In some situations, the graph $G$ could be a non-complete one. For example, in \reviewd{real-world instances}, there \reviewd{would be} an
extra requirement where certain pairwise of nodes must not participate in the same clique.
This can be modeled by simply deletion of edges $\{u,v\}$ where $u$ and $v$ are those nodes that should be included in different cliques.
Such a problem will be known as the \emph{generalized balanced $k$-\reviewd{way} clique partitioning with weight constraints} (GB$k$-WWC). The last problem is not harder to solve than B$k$-WWC. In fact, given an instance of GB$k$-WWC defined by an arbitrary graph $G=(V,E)$, cost on the edges $d: E \longrightarrow \RR^+$, weights on nodes $w:V \longrightarrow \RR^+$, positive numbers $W_L$ and $W_U$, and a fixed integer number $k \geq 2$, an instance of B$k$-WWC can be constructed as follows: take a complete graph $G' = (V', E')$ with set of nodes
$V'=V$, weights $w'_i=w_i ~~\forall i\in V$, numbers $W'_L=W_L$, $W'_U=W_U$ , $k'=k$, and distances
$$
d^{\prime}_{ij}=
\begin{cases}
d_{ij}, \mbox{ if } \{i,j\} \in E\\
M, \mbox{ if } \{i,j\} \in E' \backslash E
\end{cases}
$$
where $M$ is a big number.
It is then obvious that GB$k$-WWC has an optimum solution if and only if the optimum of B$k$-WWC does not
exceed $M$. In practice, one does not have to deal with $M$. As in the case for dummy nodes, considering a constraint $\sum_{\{i,j\} \in E' \backslash E} x_{ij} = 0$ is enough.

\section{Valid inequalities}
\label{valid_inequalities}
\reviewd{Let $\mathcal{P}$ be the polytope defined by the convex hull of integer solutions of $\F_1$ (if $n\!\mod k = 0$) or $\F_2$ (otherwise).} $\mathcal{P}$ is uniquely determined by the parameters $n, k, W_L, W_U$ and $w_i$ for all $i\in V$.
If the weight constraints are redundant, the dimension of this polytope is given by Theorem 3.1 of \citet{LABBE} and the equations (\ref{restr_2x}) or (\ref{restr_5}) are enough to describe the minimal system of $\mathcal{P}$.
On the other hand, the weight constraints could make this polytope to be empty. In order to avoid these cases, a necessary condition on weights is established in the following result:
\begin{lemma}
\label{condicion_modelo}
A necessary condition for the feasibility of B$k$-WWC is
\begin{eqnarray*}
\max \left \{ 2, \left\lceil \frac{\sum_{i\in V}w_i}{W_U} \right\rceil \right \} \leq k \leq \min \left \{ \left \lfloor \frac{n}{2} \right \rfloor , \left\lfloor \frac{\sum_{i\in V}w_i}{W_L}\right\rfloor \right\}
\end{eqnarray*}
\end{lemma}
\begin{proof}
Observe that by assumption $k\geq 2$ and $F_L\geq 2$. On the other hand, from constraints \eqref{restr_3},
\begin{align*}
& \sum_{c=1}^k W_L \leq  \sum_{c=1}^k (w_i + \sum_{\{i,j\}\in E} w_j x_{ij}) \leq \sum_{c=1}^k W_U,\\
&\sum_{c=1}^k W_L \leq \sum_{c=1}^k \sum_{i\in V_c} w_i \leq \sum_{c=1}^k W_U,\\
& k W_L \leq  \sum_{i\in V} w_i \leq k W_U,
\end{align*}
from which the result follows.
\end{proof}

Observe that \reviewd{$\mathcal{P}$} is contained in the one given by \citet{LABBE}. Hence, linear relaxations of our formulations can be strengthened by means of
known classes of valid and facet-defining inequalities described in previous works. This is the case of the 2-partition inequalities:
\begin{equation} \label{part2eq}
\sum_{i\in S} \sum_{j\in T} x_{ij} - \sum_{i_1, i_2\in S} x_{i_1 i_2} -  \sum_{j_1, j_2\in T} x_{j_1 j_2} \leq |S|,
\end{equation}
for every two disjoint nonempty subsets $S, T$ of $V$ such that $|S| \leq |T|$ and $S \cap T = \emptyset$.
These inequalities were introduced in the nineties by \citet{GROTSCHEL} for the Clique Partitioning Polytope and, in recent years, \citet{LABBE} explored these inequalities for their polytope.
Based on the computational experiments reported in these preceding works, the usage of 2-partition inequalities as cuts evidenced a good behavior and effectiveness in solving the IP model,
and they will be considered in the present paper.

\reviewd{In addition, new valid equations and inequalities arise by the introduction of weight constraints \eqref{restr_3}.
Below, three families of valid inequalities for $\mathcal{P}$ are proposed.
For any $T \subset V$, define $w(T) = \sum_{i \in T} w_i$. Also, for a given $k$-partition $\pi=(V_1,\ldots,V_k)$
define $E_{\pi}(T) = \cup_{i=1}^k E(V_i \cap T)$. That is, $E_{\pi}(T)$ contains all the edges of $\pi$ with end nodes in $T$.
Finally, any integer solution lying in a polyhedron $P$ is called a \emph{root} of $P$.
\begin{proposition}
Let $T \subset V$ such that $w(T) > W_U$. Then, the following $T$-Weight-Cover inequality is valid for $\mathcal{P}$:
\[ \sum_{\{i,j\} \in E(T)} x_{ij} \leq \dfrac{(|T|-1)(|T|-2)}{2}. \]
\end{proposition}
\begin{proof}
Let $x$ be a root of $\mathcal{P}$ representing a $k$-partition $\pi$.
The left side of the inequality is $|E_{\pi}(T)| = \sum_{i=1}^k |V_i \cap T|(|V_i \cap T|-1)/2$.
Since $w(T) > W_U$,  all nodes in $T$ do not belong to the same clique in $\pi$. That is, $T$ has nodes from two or more cliques in $\pi$
and the largest value of $|E_{\pi}(T)|$ is reached when $T$ has $|T|-1$ nodes belonging to a clique, say $V_{i_1}$, and just one node belonging
to another one, say $V_{i_2}$.
In that case, $|E(V_{i_1})| = (|T|-1)(|T|-2)/2$ and $|E(V_{i_2})| = 0$. Therefore,
$|E_{\pi}(T)| \leq (|T|-1)(|T|-2)/2$.
\end{proof}
\begin{corollary} \label{ZEROW}
Let $\{i,j\} \in E$ such that $w_i + w_j > W_U$. Then, $x_{ij} = 0$ is a valid equation of $\mathcal{P}$.
\end{corollary}
\begin{proposition}
Let $T \subset V$ such that $w(T) > W_L$, $|T| \leq F_L$ and $r = w(T) - W_L$.
Then, for all $i\in T$, the following $(T,i)$-Weight-Lowerbound inequality is valid for $\mathcal{P}$:
\[ w_i + \sum_{j \in T \setminus \{i\}} w_j x_{ij} + \sum_{j \in V \setminus T} (w_j + r) x_{ij} \geq w(T). \]
\end{proposition}
\begin{proof}
Let $x$ be a root of $\mathcal{P}$ representing a $k$-partition $\pi$ and w.l.o.g.$\!$ suppose that $i \in V_1$.
The left side of the inequality is $w(V_1) + r |V_1 \setminus T|$.
If $V_1 \setminus T = \emptyset$, we have $T = V_1$ since $|T| \leq F_L \leq |V_1|$, and the inequality is valid.
If $V_1 \setminus T \neq \emptyset$ then $w(V_1) + r |V_1 \setminus T| \geq w(V_1) + r \geq W_L + r = w(T)$.
\end{proof}
\begin{proposition}
Let $T \subset V$ such that $w(T) < W_U$,
$r = W_U - w(T)$ and $S = \{j \in V \setminus T : w_j > r\} \neq \emptyset$.
Then, for all $i\in T$, the following $(T,i)$-Weight-Upperbound inequality is valid for $\mathcal{P}$:
\[ w_i + \sum_{j \in T \setminus \{i\}} w_j x_{ij} + \sum_{j \in S} (w_j - r) x_{ij} \leq w(T). \]
\end{proposition}
\begin{proof}
Let $x$ be a root of $\mathcal{P}$ representing a $k$-partition $\pi$ and w.l.o.g.$\!$ suppose that $i \in V_1$.
The left side of the inequality is $w(T \cap V_1) + w(S \cap V_1) - r |S \cap V_1|$.
If $S \cap V_1 = \emptyset$, we have $w(T \cap V_1) \leq w(T)$ and the inequality is valid.
If $S \cap V_1 \neq \emptyset$, then
$w(T \cap V_1) + w(S \cap V_1) - r |S \cap V_1| \leq w(V_1) - r |S \cap V_1| \leq W_U - r |S \cap V_1| \leq W_U - r = w(T)$.
\end{proof}
}
\reviewd{
The previous results just give conditions for the inequalities to be valid but we are also interested in finding those inequalities
that define faces of high dimension, preferably facets of $\mathcal{P}$, since one can reinforce linear relaxations with them.
This fact would require a deeper polyhedral study of $\mathcal{P}$. However, for practical purposes, it is enough to propose necessary
conditions that guarantee that faces defined by such inequalities have as many roots as possible (or at least be non-empty).
These conditions can be further used for the proper separation of the inequalities involved.
\begin{proposition}
Let $\mathcal{F}$ be the face defined by a $T$-Weight-Cover inequality.
If $\mathcal{F} \neq \emptyset$, then there exists $l \in T$ such that $w(T \setminus \{l\}) \leq W_U$.
\end{proposition}
\begin{proof}
Let $x$ be a root of $\mathcal{F}$ representing a $k$-partition $\pi$.
Since $E_{\pi}(T) = (|T|-1)(|T|-2)/2$, the partition $\pi$ restricted to $T$ must have two components: a
clique of size $|T|-1$ and an isolated node $l$. Denote $T' = T \backslash \{l\}$. Clearly, $E_{\pi}(T) = E_{\pi}(T')$.
Now, suppose that $w(T') > W_U$. Hence, the $T'$-Weight-Cover inequality is also valid and we obtain
$E_{\pi}(T) = E_{\pi}(T') \leq (|T|-2)(|T|-3)/2 < (|T|-1)(|T|-2)/2$ which is absurd.
\end{proof}
The previous result suggests that, in order to obtain a good Weight-Cover inequality, we should impose that
$w(T \setminus \{l\}) \leq W_U$ for all $l \in T$ (\ie $T$ is ``minimal'' with respect to the condition $w(T) > W_U$).
\begin{proposition}
Let $\mathcal{F}$ be the face defined by a $(T,i)$-Weight-Lowerbound inequality.
If $\mathcal{F} \neq \emptyset$, then $|T| \geq F_L-1$.
\end{proposition}
\begin{proof}
Let $x$ be a root of $\mathcal{F}$ representing a $k$-partition $\pi$ and suppose that $i \in V_1$.
Also, let $s = |V_1 \setminus T|$. We have $w(V_1) + rs = w(T)$.
Therefore,
$$s = \dfrac{w(T) - w(V_1)}{r} = \dfrac{w(T) - w(V_1)}{w(T) - W_L}.$$
Since $w(V_1) \geq W_L$, we have $s \leq 1$. Therefore, $|T| \geq |V_1| - s \geq |V_1| - 1 \geq F_L - 1$.
\end{proof}
This result simply suggests to consider only Weight-Lowerbound inequalities such that $|T| \geq F_L-1$.
}

Regarding the Weight-Upperbound inequalities, and for the sake of clarity, 
the roots of the faces defined by such inequalities are classified in two types.
Let $\mathcal{F} \neq \emptyset$ be the face defined by a $(T,i)$-Weight-Upperbound inequality, let $x$ be a root of $\mathcal{F}$
representing a $k$-partition $\pi$ where w.l.o.g. $i \in V_1$.
If $S \cap V_1 = \emptyset$, we say that $x$ is of \emph{Type 1}. Otherwise, $x$ is of \emph{Type 2}.
Now, define $\Pi^t$ as the set of roots of $\mathcal{F}$ of Type $t$ where $t \in \{1, 2\}$.
\begin{proposition}
Let $\mathcal{F}$ be the face defined by a $(T,i)$-Weight-Upperbound inequality.\\
(i) If $\mathcal{F} \cap \Pi^1 \neq \emptyset$, then $|T| \leq F_U$.\\
(ii) If $\mathcal{F} \cap \Pi^2 \neq \emptyset$, then $|T| \geq F_L-1$.
\label{Proposition6}
\end{proposition}
\begin{proof}
We recall that $w((T \cup S) \cap V_1) - r |S \cap V_1| = w(T)$.\\
(i) If $S \cap V_1 = \emptyset$, we obtain $w(T \cap V_1) = w(T)$ implying that $T \subset V_1$.
Therefore, $|T| \leq |V_1| \leq F_U$.\\
(ii) If $S \cap V_1 \neq \emptyset$, add $r |S \cap V_1| - W_U$ to both sides of the equation. We obtain
$w((T \cup S) \cap V_1) - W_U = w(T) - W_U + r |S \cap V_1| = -r + r |S \cap V_1| = r(|S \cap V_1| - 1)$.
Since $|S \cap V_1| \geq 1$, $w((T \cup S) \cap V_1) - W_U \geq 0$ and, therefore, $w((T \cup S) \cap V_1) \geq W_U$.
On the other hand, $w(V_1) \leq W_U$ implying that $w((T \cup S) \cap V_1) = W_U$.
Hence, $w(V_1) = W_U$, $V_1 \subset T \cup S$ and $|S \cap V_1| = 1$.
Let $s$ be the unique element from $S \cap V_1$. Then, $V_1 \setminus \{s\} \subset T$.
Therefore, $|T| \geq |V_1|-1 \geq F_L-1$.
\end{proof}
This result suggests to discard those Weight-Upperbound inequalities such that the condition $F_L-1 \leq |T| \leq F_U$ does not
hold.

\section{Deriving upper bounds: a tabu search} \label{SectionHeuristic}

Consider an optimization problem where, given a graph $G = (V, E)$ and a number $k$, the objective is to obtain a partition
$(V_1, \ldots, V_k)$ of the set of nodes such that $||V_i|-|V_j|| \leq 1$ for all $i \neq j$, and to minimize the number of edges
in $\bigcup_{i=1}^k E(V_i)$. This problem, called $k$-ECP, is iteratively used by the state-of-the-art tabu search algorithm
\textsc{TabuEqCol} for solving the \emph{Equitable Coloring Problem} \citep{TABUEQCOL,TABUIMPROVED}.

The $k$-ECP is very related to the problem presented in this paper. In fact, it is a particular case of B$k$-WWC: simply
consider a complete graph $G' = (V, E')$, $W_L = W_U = 0$, $w_i = 0$ for all $i \in V$, $d_{ij} = 1$ for all $\{i,j\} \in E$ and
$d_{ij} = 0$ for all $\{i,j\} \in E' \backslash E$.
\reviewd{In this section, we propose a two-phase algorithm based on \textsc{TabuEqCol} for solving B$k$-WWC which incorporates an additional
mechanism to deal with weights.}

\emph{Tabu search} is a metaheuristic method proposed by \citet{GLOVER}. Basically, it is a local search algorithm which is equipped with
additional mechanisms that prevent from visiting a solution twice and getting stuck in a local optimum.
The design of a tabu search algorithm involves to define the search space of feasible solutions, an objective function, the neighborhood of
each solution, the stopping criterion, the aspiration criterion, the features and how to choose one of them to be stored in the tabu list and
how to compute the tabu tenure. The reader is referred to the work by \citet{TABUEQCOL} for the definitions of these concepts and how to denote them.

Below, the details of our algorithm are presented:
\begin{itemize}
\item \emph{Search space of solutions}. A solution $s$ is a partition $(V_1, \ldots, V_k)$ of the set of nodes
such that $||V_i|-|V_j|| \leq 1$ for all $i \neq j$.
For the sake of efficiency, solutions are stored in memory as tuples $(V_1, \ldots, V_k, W_1, \ldots, W_k, R^+, R^-)$
where $W_i = \sum_{v \in V_i} w_v$, $R^+ = \{ i : |V_i| = \lfloor n/k \rfloor+1\}$ and $R^- = \{ i : |V_i| = \lfloor n/k \rfloor\}$.
\item \emph{Objective function}. For a given solution $s$, let $d(s)$ be the sum of the distances of every edge in $E(V_i)$ for all
$i \in \{1,\ldots,k\}$ and $I(s) = \{ i : W_i < W_L \lor W_i > W_U\}$. The objective function is defined as $f(s) = d(s) + M|I(s)|$
where $M$ is a big value. Note that solutions satisfying $I(s) \neq \varnothing$ are feasible but penalized in $f(s)$.
\item \emph{Stopping criterion}. The algorithm stops when a maximum number of iterations is reached.
\item \emph{Aspiration criterion}. Let $s$ be the current solution and $s^*$ be the best known solution so far. When $f(s) < f(s^*)$,
$s$ replaces $s^*$.
\item \emph{Set of features}. A solution $s$ presents a feature $(v, i)$ if and only if $v \in V_i$.
\item \emph{Initial solution}. For all $i \in \{1,\ldots,k\}$, do $V_i = \{ v_j : (j-1) \mod k = i-1 \}$.
\item \emph{Neighborhood of a solution}. For a given solution $s = (V_1, \ldots, V_k$, $W_1, \ldots, W_k$, $R^+, R^-)$, a neighbor
$s'  = (V'_1, \ldots, V'_k$, $W'_1, \ldots, W'_k$, $R'^+, R'^-)$ of $s$ is generated with two schemes:
\begin{itemize}
\item \emph{1-move} (only applicable when $n$ does not divide $k$). For a given $v \in V_i$ such that $i \in R^+$ and a given $j \in R^-$,
consider $s'$ such that node $v$ is moved from $V_i$ to $V_j$. Formally,
$V'_j = V_j \cup \{ v\}$, $W'_j = W_j + w_v$, $V'_i = V_i \backslash \{v\}$, $W'_i = W_i - w_v$,
$V'_l = V_l$ and $W'_l = W_l$ for all $l \in \{1,\ldots,k\} \backslash \{i, j\}$,
$R'^+ = R^+ \cup \{j\} \backslash \{i\}$ and $R'^- = R^- \cup \{i\} \backslash \{j\}$.
\item \emph{2-exchange}. For a given $v \in V_i$ and $u \in V_j$ such that $i < j$, consider $s'$ such that
$v$ is moved to $V_j$ and $u$ is moved to $V_i$. Formally,
$V'_j = (V_j \backslash \{u\}) \cup \{v\}$, $W'_j = W_j - w_u + w_v$,
$V'_i = (V_i \backslash \{v\}) \cup \{u\}$, $W'_i = W_i - w_v + w_u$,
$V'_l = V_l$ and $W'_l = W_l$ for all $l \in \{1,\ldots,k\} \backslash \{i, j\}$,
$R'^+ = R^+$ and $R'^- = R^-$.
\end{itemize}
\item \emph{Selection of a feature to add in the tabu list}. Once a movement from $s$ to $s'$ is performed,
$(v,i)$ is stored on the tabu list.
\item \emph{Tabu tenure}. Once an element is added to the tabu list, it remains there for the next $t$ iterations,
where $t$ is an integer randomly generated with a uniform distribution (one of the criteria used by \citet{TABUIMPROVED}).
\end{itemize}
Since one pretends the algorithm to be as fast as possible, the value of $f(s')$ should be computed by adding or subtracting
the corresponding difference to $f(s)$. Also, $M$ should not be too high in order to avoid round-off errors. In our case, $M$
was set to $10000$.

\reviewd{A difference between \textsc{TabuEqCol} and our algorithm is that we are interested in feasible solutions for the B$k$-WCC
but \textsc{TabuEqCol} does not contemplate weight constraints. For that reason, the entire process is divided in two stages.
The first one consists of searching a solution $s$ with $I(s) = \varnothing$ while the second one is focused on minimizing $d(s)$.}
We observed that, during the first stage, the search needs to be more diversified.
Therefore,  different range of values for the tabu tenure in each stage are used.
If $f(s) \geq M$ (1st. stage) then $t \in [5,40]$ and if $f(s) < M$ (2nd. stage) then $t \in [5,20]$.

This method can also be used for obtaining feasible solutions of GB$k$-WWC: \reviewd{simply consider $d_{ij} = M$ for those edges
$\{i,j\} \notin E$; if $f(s) < M$ then $s$ is a feasible solution of GB$k$-WWC.}
However, if the density of edges in the graph $G$ is not high,
it would be convenient to exploit the structure of $G$ in the computation of the neighborhood of a solution
and tabu tenure, as in the case of \textsc{TabuEqCol}.

\section{Computational experiments} \label{SectionPracticalExperience}

In this section, some computational experiments are presented. They consist of comparing different ways of solving the B$k$-WWC, called
\emph{Strategies} and denoted by $\mathcal{S}$. Comparisons are carried out over random instances and, at the end, the resolution of a real-world instance is addressed.
The improvement in the results are shown incrementally. That is, each strategy outperforms the previous one.
All the experiments were carried out on a machine equipped with an Intel i5 2.67GHz, 4Gb of memory, the operating system Ubuntu 16.4 and
the solver GuRoBi 6.5. All instances as well as the source code of the implementation can be downloaded from:

\begin{center}
\texttt{http://www.fceia.unr.edu.ar/$\sim$daniel/stuff/football.zip}
\end{center}

Random instances were generated by computing the coordinates $(x,y)$ of $n$ points with an uniform distribution in the domain
$[-100, 100] \times [-100, 100]$. Then, for every pair of points $i$, $j$, $d_{ij}$ is assigned the euclidean distance between both points.
Weights $w_i$ are random values generated with a uniform distribution in the range $[0.1, 0.9]$,
and $W_L = \mu (n / k) - \sigma$, $W_U = \mu (n / k) + \sigma$ where $\mu$ and $\sigma$ are the average and standard deviation of the
weights of all points.
Combinations of $n$ are $k$ were chosen so that $k$ does not divide $n$, similar to those values of real
instances.

Results of the experiments are summarized in Tables \ref{tab:1} to \ref{tab:5}, whose
format are as follows: first and second columns display the number of the instance and its optimal
value, and remaining columns show the number of $B\&B$ nodes evaluated and the time elapsed in seconds of
the optimization. A time limit of one hour is imposed. For those instances that are not solved within
this limit, the gap percentage between the best lower and upper bound found is reported.
Last row displays the average over all instances, except for Tables \ref{tab:3} and \ref{tab:5} (marked with a
dagger) where it shows the average over those instances solved by all strategies being compared
(\ie 21, 23, 24, 26, 28, 29 and 30 for Table \ref{tab:3}; 42, 43, 45, 47, 49 and 50 for Table \ref{tab:5}).
Boldface indicates the best results.\\

\reviewd{\emph{$\F_1$ (with dummy nodes) vs. $\F_2$}.}
Strategy 1 ($\mathcal{S}_1$) consists of the resolution of $\F_1$, after addition of $k-(n\!\mod k)$ dummy
nodes, while $\mathcal{S}_2$ is the direct resolution of $\F_2$.
Both use GuRoBi at its default settings.
Instead of using a single constraint, such as \eqref{rest_ceros}, 
variables $x_{ij}$ are directly fixed to zero in a straightforward manner.
In fact, there are three cases where $x_{ij}$ are set to zero:
\begin{itemize}
\item  $i, j \in \DUMMY$ (only when dummy nodes are present, see constraint \eqref{rest_ceros}).
\item $\{i,j\} \in E' \backslash E$ (only when $G$ is not complete).
\item $w_i + w_j > W_U$ (see Corollary \ref{ZEROW}).
\end{itemize}

Note that, according to the results reported in the tables, $\mathcal{S}_2$ performs better than $\mathcal{S}_1$ for larger instances. In particular,
$\mathcal{S}_2$ solves instance 27 to optimality and reports better gaps than
$\mathcal{S}_1$ for instances 22 and 25.\\

\reviewd{\emph{Tabu search vs. GuRoBi built-in heuristics}.}
The next step is to use the tabu search method proposed in the previous section. This metaheuristic generates
good initial solutions. In particular, it gives the optimal one for almost all instances given
in the tables in a reasonable amount of iterations (the unique exception was instance 19 where it could not reach the optimal solution after 1000000 iterations, for two different seeds). Based on experimentation with
several random instances and initial seeds, we obtained a formula by linear regression for the limit in the number of iterations needed:
\[ it_{limit} = \left\lfloor e^{0.26571 n - 0.052978} \right\rfloor \]
Now, in $\mathcal{S}_3$, $it_{limit}$ iterations of the tabu search are executed and the best solution found is provided as an initial integer solution to GuRoBi.
Then, $\F_2$ is solved with GuRoBi primal heuristics turned off (Heuristics, PumpPasses and RINS parameters are set to zero).
In order to make a fair comparison, time reported in tables includes time spent by tabu search.
Clearly, $\mathcal{S}_3$ outperforms $\mathcal{S}_2$. In particular, $\mathcal{S}_3$ was able to solve
instance 25 by optimality and presents a lower gap than $\mathcal{S}_2$ for instance 22.\\

\reviewd{\emph{Addition of triangular inequalities on demand}.}
Formulation $\F_2$ (and also $\F_1$) has a $O(n^3)$ number of constraints due to triangular inequalities
(precisely, $n(n-1)(n-2)/2$). Although in presence of variables set to zero some of them become redundant
and one can omit them when performing the initial relaxation
(e.g.\! for a given $i < j < l$, if $x_{ij}=0$ then inequalities \eqref{restr_11} and \eqref{restr_12}
are redundant but \eqref{restr_13} is not), its number is still high.
Since the number of variables is $O(n^2)$ it is expected that several triangular inequalities are not
active in the fractional solution of the initial relaxation. We noticed that there exists a relationship between being active and the distances of positive variables in its support: the lower the value is
$d_{ij} + d_{jl}$, the higher the probability of the inequality $x_{ij} + x_{jl} - x_{il} \leq 1$ is
active. The following experiment reveals this relationship.

\reviewd{
Let $\mathcal{T}$ be the set of triangular inequalities and let $\tilde d(t)$ be a value assigned to each $t \in \mathcal{T}$
as follows:
$$\tilde d(t) = \begin{cases}
d_{ij} + d_{jl}, &\textrm{when $t$ is a constr. \eqref{restr_11}},\\
d_{ij} + d_{il}, &\textrm{when $t$ is a constr. \eqref{restr_12}},\\
d_{jl} + d_{il}, &\textrm{when $t$ is a constr. \eqref{restr_13}}. \end{cases}$$
We first order all triangular inequalities according to the value $\tilde d(t)$ from lowest to highest and we make an equitable partition
$I_1, I_2, \ldots, I_{10}$ of the set $\mathcal{T}$ in 10 deciles. That is, $\cup_{r=1}^{10} I_r = \mathcal{T}$ and
$|I_r| = \lceil\frac{n+r-10}{10}\rceil$ for all $r = 1,\ldots,10$ where $d(t) \leq d(t')$ for $t \in I_r$ and $t' \in I_{r+1}$.}
Then, we solve $\F_2$ without these inequalities in its initial relaxation and, whenever GuRobi obtains a solution (fractional or integer) violating some of them,
they are added to the current relaxation: if the solution is fractional, they are added as cuts and if the solution is integer, they are added as lazy constraints.
Histograms with the averages (over 10 instances of 44 nodes, each instance having $|\mathcal{T}|=39732$) of percentages of triangular
inequalities generated per decile $I_r$ are shown in Figures \ref{fig:1} and \ref{fig:2}. In the former, only those inequalities added at
root node of $B\&B$ tree are considered. In Figure \ref{fig:2}, all inequalities are counted. In particular, if the same inequality is
generated in two different nodes, it is counted twice.

Note that, at root node, most of the inequalities from $I_1$ and $I_2$ are generated. In addition, during the B\&B process, inequalities
from $I_8$, $I_9$ and $I_{10}$ are rarely violated.
Based on these observations and additional experimentation, \reviewd{we define the strategy} $\mathcal{S}_4$ as $\mathcal{S}_3$ with the following differences:
\begin{itemize}
\item Only inequalities from $I_1$ and $I_2$ are considered in the initial relaxation.
\item Each time an integer solution is found, inequalities from $I_r$ with $r \in \{3,\ldots,10\}$ are checked and violated ones are added as lazy constraints.
\item Each time a fractional solution is found, inequalities from $I_r$ with $r \in \{3,\ldots,m\}$ are checked and those that are violated by at least 0.1 units, are added as cuts. If the current node is root, $m = 10$, otherwise $m = 7$.
\item Some GuRoBi cuts are disabled (Clique, FlowCover, FlowPath, Implied, MIPSep, Network and ModK).
\end{itemize}
Observe that $\mathcal{S}_4$ behaves much better than $\mathcal{S}_3$. This strategy needs less than the half of time used by the preceding
one. In addition, it can solve instance 22 to optimality.\\

\reviewd{\emph{Separation of valid inequalities}.
Here, we experiment with additional custom separation routines embedded in our code, where 2-partition inequalities and the new families
of valid inequalities presented in Section \ref{valid_inequalities} are considered.
Two experiments are carried out, detailed below.
}

\reviewd{
\noindent Experiment 1: In this experiment, we analyze the effectiveness in terms of reduction in the number of B\&B nodes when
Weight-Cover, Weight-Lowerbound and Weight-Upperbound inequalities are used.
We also gather helpful information that is further used for the design of the separation routines.
For each family, random instances of 44 nodes are solved and, during the optimization, the inequalities satisfying the conditions given in Section \ref{valid_inequalities} are enumerated exhaustively and added when they are violated by an amount of at least $1\%$ of the r.h.s.\\
Regarding Weight-Cover inequalities, we restrict the enumeration to $|T| \in \{4,5\}$ since for $|T| = 3$ or $|T| \geq 6$ they are seldom violated.
Regarding Weight-Upperbound inequalities, we impose an additional limit of 1500 nodes since its enumeration in each node consumes a fairly
long time. This limit is reached on instances 22, 25 and 27.\\
Table \ref{tab:0} reports the total number of cuts generated and the number of B\&B nodes evaluated for each family of valid inequalities, and
the relative gap when 1500 nodes are reached (only for Weight-Upperbound).
The three columns entitled ``only GuRoBi'' display the same parameters (number of cuts, B\&B nodes and relative gap at 1500 nodes) generated by
strategy $\mathcal{S}_4$. The last three rows show the averages over all instances, the averages over instances 21, 23, 24, 26, 28, 29 and 30
(marked with a dagger), and the average of gap over instances 22, 25 and 27 (marked with a double dagger).\\
We conclude that the addition of Weight-Cover and Weight-Upperbound cuts make a substantial reduction in the number of nodes,
whereas Weight-Lowerbound cuts occur less frequently and the reduction in the number of nodes is marginal.
We also noticed that violated Weight-Cover inequalities are usually composed of nodes $i$ with high values of
the expression $w_i + \sum_{j \in V \setminus \{i\}} w_j x^*_{ij}$, where $x^*$ is the current fractional solution.
However, violated Weight-Upperbound inequalities do not seem to have an obvious structure that can be exploited.
We only consider Weight-Cover inequalities in the next experiment.
}

\reviewd{
\noindent Experiment 2: The goal is to compare the performance of $\mathcal{S}_4$ when different combinations of
separation routines are used. For each combination, random instances of 48 and 54 nodes are solved (see Tables \ref{tab:4} and \ref{tab:5}).
The execution of these routines is performed only when no triangular inequalities were generated for the current fractional solution,
denoted by $x^*$. In particular, the separation of 2-partition inequalities is based on the procedure given by \citet{GROTSCHEL} and
\citet{LABBE2}.}
\begin{itemize}
\item \emph{2-partition}.
The following procedure is repeated for each $v \in V$.
First, compute $W := \{u \in V \backslash \{v\} : x^*_{uv} \notin \mathbb{Z} \}$.
If $|W| \leq 4$, then stop. Otherwise, set $F := \emptyset$ and repeat the following 5 times, whenever possible.
Pick two random nodes $i, j$ from $W\backslash F$ and set $T := \{i,j\}$. Keep picking nodes $r$ from $W\backslash F$ such that
$x^*_{rv} - \sum_{t\in T} x^*_{rt} > 0$ and add $r$ to $T$ until no more nodes are found.
Then, check if the 2-partition inequality \eqref{part2eq} with $S=\{v\}$ and $T$ is violated and, in that case, add it as a cut.
Finally, make $F := F \cup T$ (even when the inequality is not violated).
The set $F$ (of ``forbidden nodes'') prevents from generating cuts with similar support.
\item \emph{Weight-Cover}.
First, order nodes $i \in V$ according to the value $q^*(i) := w_i + \sum_{j \in V \setminus \{i\}} w_j x^*_{ij}$ from highest to lowest,
i.e.$\!$ $V = \{v_1, v_2, \ldots, v_n\}$ such that $q^*(v_j) \geq q^*(v_{j+1})$ for all $j$.
For each $t \in \{4,5\}$ do the following.
Consider every $T$ composed of $t - 2$ nodes from $\{v_1, v_2, \ldots, v_{t+2}\}$ and 2 more nodes from $V$
(note that $|T| = t$ and the procedure explores $O(n^2)$ combinations).
If $w(T) > W_U$ and $w(T \setminus \{l\}) \leq W_U$ for all $l \in T$, check the $T$-Weight-Cover inequality.
If it is violated, add it as a cut.
\end{itemize}

\reviewd{
In the given procedures, an inequality is considered violated if the amount of violation is at least $0.1$.
As one can see from the tables, 2-partition together with Weight-Cover cuts is the best choice.
We define the strategy $\mathcal{S}_5$ as $\mathcal{S}_4$ with both separation routines enabled.\\
}

\review{\emph{Resolution of a real-world instance}. As mentioned in the introduction, in the zonal stage of the second category of the Ecuadorian football league, a championship including the two best teams of each provincial associations is played. The set of teams must be partitioned in 8 groups to play a Round Robin Tournament in each one of them.  A regulation imposes that two teams of the same provincial association must not belong to the same group. During the 2015 season,  $44$ teams (22 provincial associations) participated in the tournament, and they were divided in $4$ groups of $5$ teams and $4$ groups of $6$ teams.}
	
\review{ The nodes of a graph are associated with the venues of the teams.  We denote the nodes by $2i$ and $2i+1$, corresponding to the venues of the best two teams of each provincial association $i$, for all $i=0,1,\ldots,21$, and edge $\{i,j\}$ is included if and only if the teams associated to nodes $i,j$ could potencially belong to the same group. In order to satisfy the regulation mentioned before, edges of the form  $\{2i,2i+1\}$ do not appear in the graph. Thus, our realignment instance consists of $44$ teams which must be partitioned in $8$ groups, and the graph $G = (V,E)$ is a particular complement of a matching, \ie $V = \{0,\ldots,43\}$ and $E = \{ \{i,j\} : 0 \leq i < j \leq 43 \} \backslash \{ \{2i,2i+1\} : 0 \leq i \leq 21 \}$. 
}

For solving our instance, we made a preliminary test of our two best strategies (\ie $\mathcal{S}_4$ vs.
$\mathcal{S}_5$). A time limit of one hour was imposed to them. None of them was able to solve the instance within this limit, but the
relative gap reported was 3.22\% for $\mathcal{S}_4$ against 2.36\% for $\mathcal{S}_5$. 
Hence, $\mathcal{S}_5$ was chosen for solving the instance without time limit.
Below, we resume some highlights about the optimization:
\begin{itemize}
\item \review{\emph{Instance:} $|V|=44$, $k = 8$, $W_L = 2.08412$ and $W_U = 4.14835$. }
\item \emph{Iterations performed/time spent by tabu search}: 113352 iterations (6.8 sec.).
\item \emph{Variables and constraints of the initial relaxation:} 913 vars., 7444 constr. 
\item \emph{Cuts generated}: Gomory (18), Cover (65), MIR (1620), GUB cover (98), Zero-half (1278), triangular (1935), 2-partition (3965), Weight-Cover (14).
\item \emph{Nodes explored and total time elapsed}: 34573 (68374 sec.).
\item \emph{Optimal value}: 21523 km, found by tabu search at iter. 3549 (0.21 sec.).
\item \emph{Gap evolution:} 2.36\% after 1 hour, 1.08\% after 4 hours.
\end{itemize}
Since a Double Round Robin Tournament is considered, the total distance is 86092 km.
In contrast, the best solution found in our previous work \citep{ISCO2016} had 86192 km and the gap reported was 12.6\% after 4 hours of execution.

\section{Conclusion and future work} \label{Conclusions}

\review{In this paper, a balanced $k$-way partitioning problem  with weight constraints is defined. The problem consists in partitioning a complete or a general graph in a fixed number of subsets of nodes such that their cardinality differs at most in one and the total weight of each subset is bounded. The objective aims to minimize the total cost of edges with end-nodes in the same subset.  The problem was formulated as an integer program and several strategies based on exact and methaheuristic methods are evaluated.  The motivation to state, formulate and solve this problem was the realignment of the second category of the Ecuadorian football league. The solution of this case of study is  based on real world data and the methodology may be suitable applied to the realignment of other sports leagues with similar characteristics.}

In order to solve the problem, one of the key decisions was to use a modification of a state-of-the-art tabu search for obtaining good feasible solutions.
In particular, our algorithm found the optimal solution of the real-world instance in less than a second, whereas the previous approach given
by \citet{ISCO2016} was unable to obtain it within 4 hours of CPU time.
Another key decision was to include a portion of triangular inequalities (20\% of them under a specific ordering) in the initial relaxation
and manage the remaining ones as cuts or lazy constraints. 
These facts, in conjunction with other minor decisions, allowed to solve comfortably random instances of 48 nodes in less than half an hour
and the real-world one in 19 hours (here, almost all the time was spent on certifying the optimality). 

In addition, two formulations are presented and one of them ($\F_2$) is chosen based on experimentation.
A possible cause of the poor performance of $\F_1$ could be the existence of symmetrical solutions due to the indistinguishability
of dummy nodes.
For example, an addition of 4 dummy nodes implies that, for each integer solution of $\F_2$, there are 24 symmetrical integer solutions in $\F_1$.

\reviewd{
We also proposed three families of valid inequalities and two of them have proven to be very effective for reducing the number of B\&B nodes, when they are used as cuts.
One of them (Weight-Cover) in conjunction with the well-known 2-partition inequalities, allows to shorten the CPU time in 51\%
for $n = 48$ (see Table \ref{tab:4}) and 56\% for $n = 54$ (see Table \ref{tab:5}). Moreover, it is able to solve one more instance (48) and
the gap reported for those instances not solved in one hour (41, 44, 46) is smaller.
}\review{As other state-of-the-art exact algorithms for the $k$-way graph partitioning problem \cite{Fairbrother_EtAl_2017, Anjos2013}, the best strategies provided here attain optimal solutions for graphs that have around 50 nodes.}

A future work could be \reviewd{to include a separation routine of Weight-Upperbound inequalities and to explore other valid inequalities
(for example, by adapting those ones presented by \citet{LABBE2}).
Finally, at the theoretical level, it could be useful to make a polyhedral study of $\mathcal{P}$, the convex hull of integer solutions of
$\F_2$, and to propose facet-defining inequalities that can be used as cuts.
}

\begin{acknowledgements}
This research was partially supported by the 15-MathAmSud-06
``PACK-COVER: Packing and covering, structural aspects'' trilateral cooperation project.
\end{acknowledgements}

\bibliographystyle{spbasic}
\bibliography{biblio}

\newpage

\begin{figure}
\includegraphics[scale=0.55]{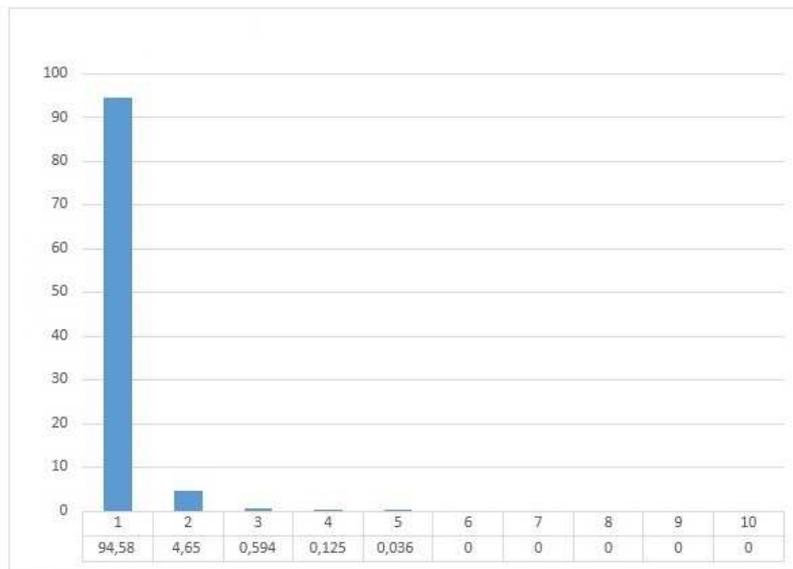}
\caption{Average of \% of triangular inequalities added at root node}
\label{fig:1}
\end{figure}

\begin{figure}
\includegraphics[scale=0.55]{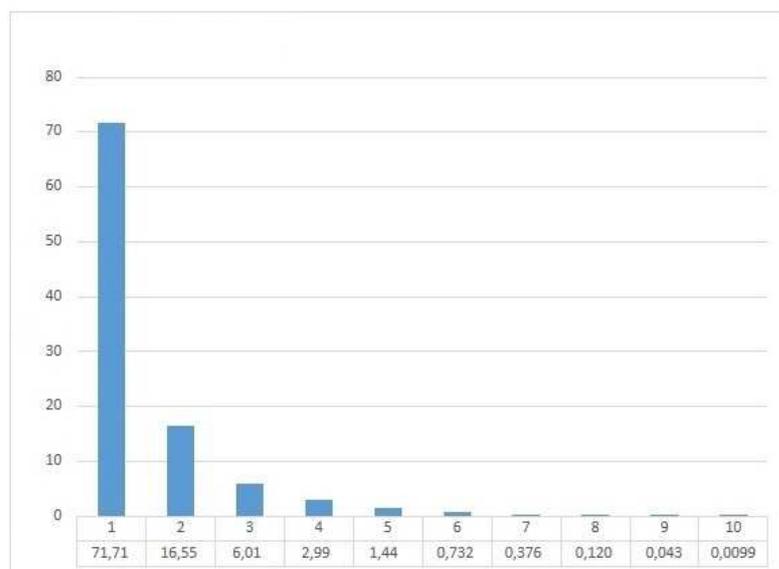}
\caption{Average of \% of total triangular inequalities added}
\label{fig:2}
\end{figure}

\begin{table} \footnotesize
\caption{Exhaustive generation of new cuts ($n=44$ and $k=8$)}
\label{tab:0}
\begin{tabular}{|c|c@{ }c@{ }c||c@{ }c|c@{ }c|c@{ }c@{ }c|}
\hline
     & \multicolumn{3}{c||}{only GuRoBi} & \multicolumn{2}{c|}{Weight-Cover} & \multicolumn{2}{c|}{W.-Lowerb.} & \multicolumn{3}{c|}{Weight-Upperbound} \\
N. & Nodes & Cuts & Gap1500 & Nodes & Cuts & Nodes & Cuts & Nodes & Cuts & Gap1500 \\
\hline
21 & 322 & 280 & & 255 & 34 & 281 & 17 & 415 & 58 &   \\
22 & 11037 & 782 & 4.71\% & 10290 & 146 & 12595 & 46 &  &  & 3.58\% \\
23 & 545 & 365 & & 519 & 69 & 629 & 16 & 592 & 208 &  \\
24 & 1492 & 576 &  & 1023 & 82 & 1465 & 27 & 1427 & 123 &  \\
25 & 5413 & 691 & 3.65\% & 6882 & 182 & 6126 & 120 &  &  & 4.10\% \\
26 & 574 & 379 & & 317 & 36 & 584 & 5 & 665 & 149 &  \\
27 & 5360 & 555 & 5.52\% & 3120 & 148 & 5219 & 98 &  &  & 4.34\% \\
28 & 50 & 84 & & 24 & 25 & 50 & 0 & 29 & 20 &  \\
29 & 1518 & 407 & & 659 & 37 & 1606 & 49 & 631 & 77 &  \\
30 & 123 & 153 & & 57 & 18 & 75 & 5 & 84 & 73 &  \\
\hline
Av. & 2643.4 & 427.2 & & 2314.6 & 77.7 & 2863.0 & 38.3 &  &  &  \\
Av.$^{\dagger}$ & 660.6 & 320.6 &  &  &  &  &  & 549.0 & 101.1 &  \\
Av.$^{\dagger\dagger}$ &  &  & 4.63\% &  &  &  &  &  &  & 4.01\% \\
\hline
\end{tabular}
\end{table}

\begin{table}
\caption{Random instances with $n=34$ and $k=6$}
\label{tab:1}
\begin{tabular}{|c|c||c@{ }c|c@{ }c|c@{ }c|c@{ }c|}
\hline
       & Opt.  & \multicolumn{2}{c|}{Strategy 1} & \multicolumn{2}{c|}{Strategy 2} & \multicolumn{2}{c|}{Strategy 3} & \multicolumn{2}{c|}{Strategy 4} \\
N. & Value & Nodes & Time & Nodes & Time & Nodes & Time & Nodes & Time \\
\hline
1 & 3669.6 & 386 & 81 & 353 & 85 & 353 & 47 & 334 & 15 \\
2 & 3653.6 & 760 & 134 & 989 & 357 & 453 & 33 & 485 & 18 \\
3 & 3559.3 & 49 & 16 & 179 & 28 & 25 & 12 & 37 & 4 \\
4 & 3950.8 & 542 & 125 & 695 & 137 & 471 & 58 & 472 & 24 \\
5 & 3973.5 & 149 & 47 & 152 & 43 & 191 & 34 & 191 & 10 \\
6 & 3664.5 & 217 & 37 & 492 & 71 & 267 & 28 & 81 & 6 \\
7 & 3487.8 & 308 & 55 & 641 & 124 & 311 & 38 & 251 & 14 \\
8 & 3293.1 & 297 & 36 & 197 & 26 & 220 & 23 & 228 & 8 \\
9 & 3937.4 & 159 & 67 & 281 & 58 & 243 & 41 & 197 & 14 \\
10 & 3841.7 & 644 & 141 & 518 & 100 & 472 & 62 & 355 & 19 \\
\hline
Av. & & 351.1 & 73.9 & 449.7 & 102.9 & 300.6 & 37.6 & \textbf{263.1} & \textbf{13.2} \\
\hline
\end{tabular}
\end{table}

\begin{table}
\caption{Random instances with $n=38$ and $k=7$}
\label{tab:2}
\begin{tabular}{|c|c||c@{ }c|c@{ }c|c@{ }c|c@{ }c|}
\hline
       & Opt.  & \multicolumn{2}{c|}{Strategy 1} & \multicolumn{2}{c|}{Strategy 2} & \multicolumn{2}{c|}{Strategy 3} & \multicolumn{2}{c|}{Strategy 4} \\
N. & Value & Nodes & Time & Nodes & Time & Nodes & Time & Nodes & Time \\
\hline
11 & 3344.4 & 188 & 100 & 30 & 33 & 80 & 29 & 18 & 7 \\
12 & 3456.8 & 498 & 272 & 618 & 117 & 115 & 23 & 124 & 8 \\
13 & 3361.7 & 284 & 133 & 213 & 56 & 278 & 46 & 220 & 14 \\
14 & 3646.4 & 709 & 286 & 395 & 68 & 384 & 45 & 288 & 14 \\
15 & 3384.7 & 1816 & 2434 & 1614 & 250 & 1469 & 188 & 1346 & 69 \\
16 & 3538.2 & 464 & 314 & 264 & 51 & 261 & 38 & 269 & 14 \\
17 & 3366.4 & 1834 & 1707 & 2455 & 485 & 2580 & 352 & 2962 & 198 \\
18 & 3565.2 & 360 & 218 & 213 & 42 & 134 & 32 & 67 & 11 \\
19 & 3029 & 100 & 73 & 301 & 39 & 317 & 37 & 148 & 11 \\
20 & 3846.6 & 1399 & 915 & 967 & 165 & 968 & 129 & 1038 & 63 \\
\hline
Av. & & 765.2 & 645.2 & 707 & 130.6 & 658.6 & 91.9 & \textbf{648} & \textbf{40.9} \\
\hline
\end{tabular}
\end{table}

\begin{table}
\caption{Random instances with $n=44$ and $k=8$}
\label{tab:3}
\begin{tabular}{|c|c||c@{ }c|c@{ }c|c@{ }c|c@{ }c|}
\hline
       & Opt.  & \multicolumn{2}{c|}{Strategy 1} & \multicolumn{2}{c|}{Strategy 2} & \multicolumn{2}{c|}{Strategy 3} & \multicolumn{2}{c|}{Strategy 4} \\
N. & Value & Nodes & Time & Nodes & Time & Nodes & Time & Nodes & Time \\
\hline
21 & 3706.1 & 624 & 799 & 445 & 145 & 421 & 127 & 322 & 46 \\
22 & 3896 & \multicolumn{2}{c|}{4.65\%} & \multicolumn{2}{c|}{3.37\%} & \multicolumn{2}{c|}{2.55\%} & 11037 & 1721 \\
23 & 3692.7 & 936 & 977 & 1058 & 302 & 683 & 157 & 545 & 53 \\
24 & 3778.7 & 1756 & 2235 & 1933 & 673 & 1784 & 508 & 1492 & 162 \\
25 & 4228 & \multicolumn{2}{c|}{3.32\%} & \multicolumn{2}{c|}{0.79\%} & 5083 & 2553 & 5413 & 894 \\
26 & 3844.3 & 908 & 808 & 856 & 276 & 677 & 171 & 574 & 67 \\
27 & 3731.2 & \multicolumn{2}{c|}{4.06\%} & 5593 & 2318 & 3943 & 1654 & 5360 & 721 \\
28 & 3567.7 & 53 & 60 & 114 & 38 & 43 & 49 & 50 & 21 \\
29 & 4205.5 & 1379 & 2576 & 1312 & 866 & 1473 & 835 & 1518 & 262 \\
30 & 3893.2 & 156 & 175 & 165 & 53 & 135 & 68 & 123 & 33 \\
\hline
Av.$^{\dagger}$ &  & 830.3 & 1090 & 840.4 & 336.1 & 745.1 & 273.6 & \textbf{660.6} & \textbf{92} \\
\hline
\end{tabular}
\end{table}

\begin{table}
\caption{Random instances with $n=48$ and $k=9$}
\label{tab:4}
\begin{tabular}{|c|c||c@{ }c|c@{ }c|c@{ }c|c@{ }c|}
\hline
       & Opt.  & \multicolumn{2}{c|}{Strategy 4} & \multicolumn{2}{c|}{+ 2-partition} & \multicolumn{2}{c|}{+ Weight-Cover} & \multicolumn{2}{c|}{+ Both} \\
N. & Value & Nodes & Time & Nodes & Time & Nodes & Time & Nodes & Time \\
\hline
31 & 3909.5 &	2985 & 386 & 1857 & 311 & 2773 & 368 & 1665 & 298 \\
32 & 3884.7	&	213 & 52 & 243 & 58 & 265 & 51 & 232 & 52 \\
33 & 3792.7	& 3376 & 329 & 2372 & 296 & 3221 & 295 & 2296 & 272 \\
34 & 4264.9	&	1945 & 622 & 2092 & 659 & 1986 & 478 & 2097 & 588 \\
35 & 3785.3	&	525 & 113 & 453 & 115 & 449 & 80 & 425 & 80 \\
36 & 3701.5	&	217 & 87 & 198 & 86 & 243 & 64 & 158 & 57 \\
37 & 4000.5	&	2542 & 729 & 2406 & 512 & 1234 & 234 & 1633 & 267 \\
38 & 4885.9	& 10132 & 2038 & 6226 & 2055 & 7294 & 1652 & 4902 & 1233 \\
39 & 3895.8	&	4980 & 1295 & 2784 & 1034 & 3956 & 852 & 3552 & 884 \\
40 & 3627.2	&	122 & 65 & 115 & 49 & 100 & 45 & 72 & 43 \\
\hline
Av. & & 2703.7 & 571.6 & 1874.6 & 517.5 & 2152.1 & 411.9 & \textbf{1703.2} & \textbf{377.4} \\
\hline
\end{tabular}
\end{table}

\begin{table}
\caption{Random instances with $n=54$ and $k=10$}
\label{tab:5}
\begin{tabular}{|c|c||c@{ }c|c@{ }c|c@{ }c|c@{ }c|}
\hline
       & Opt.  & \multicolumn{2}{c|}{Strategy 4} & \multicolumn{2}{c|}{+ 2-partition} & \multicolumn{2}{c|}{+ Weight-Cover} & \multicolumn{2}{c|}{+ Both} \\
N. & Value & Nodes & Time & Nodes & Time & Nodes & Time & Nodes & Time \\
\hline
41 & & \multicolumn{2}{c|}{2.33\%} & \multicolumn{2}{c|}{2.29\%} & \multicolumn{2}{c|}{2.34\%} & \multicolumn{2}{c|}{1.77\%} \\
42 & 3891.9 & 1673 & 798 & 1770 & 889 & 964 & 782 & 1931 & 767 \\
43 & 4153	& 7688 & 3221 & 6336 & 3309 & 7490 & 3033 & 6409 & 2402 \\
44 & & \multicolumn{2}{c|}{2.06\%} & \multicolumn{2}{c|}{2.05\%} & \multicolumn{2}{c|}{1.56\%} & \multicolumn{2}{c|}{1.79\%} \\
45 & 4321.1 & 4115 & 1916 & 3676 & 1869 & 3305 & 838 & 2570 & 578 \\
46 & & \multicolumn{2}{c|}{1.67\%} & \multicolumn{2}{c|}{2.10\%} & \multicolumn{2}{c|}{1.65\%} & \multicolumn{2}{c|}{1.62\%} \\
47 & 4203.2 & 9081 & 3519 & 5191 & 3043 & 6524 & 2869 & 4824 & 2084 \\
48 & 4039.1	& \multicolumn{2}{c|}{1.31\%} & \multicolumn{2}{c|}{1.60\%} & \multicolumn{2}{c|}{1.87\%} & 7797 & 3173 \\
49 & 4270.9	& 3337 & 972 & 3453 & 1017 & 2865 & 841 & 2716 & 664 \\
50 & 3916.5	& 5103 & 1640 & 3563 & 1654 & 3977 & 1529 & 3529 & 1247 \\
\hline
Av.$^{\dagger}$ & & 5166.2 & 2011.0 & 3998.2 & 1963.5 & 4187.5 & 1648.7 & \textbf{3663.2} & \textbf{1290.3} \\
\hline
\end{tabular}
\end{table}

\end{document}